\documentclass[10pt,conference]{IEEEtran}
\usepackage[utf8]{inputenc}
\usepackage{amsmath}
\usepackage{amsthm}
\usepackage{amsfonts}
\usepackage{mathtools}
\usepackage{multiaudience}
  \SetNewAudience{conference}
  \SetNewAudience{arxiv}

\IEEEoverridecommandlockouts

\makeatletter
\def\ps@headings{%
\def\@oddhead{\mbox{}\scriptsize\rightmark \hfil \thepage}%
\def\@evenhead{\scriptsize\thepage \hfil \leftmark\mbox{}}%
\def\@oddfoot{}%
\def\@evenfoot{}}
\makeatother
\newcommand{\lemmaconst}{\delta}
\newcommand{\channelpmf}{q}
\newcommand{\codebookpmf}{p}
\newcommand{\generalpmf}{r}
\newcommand{\generalrvOne}{A}
\newcommand{\generalrvOneValue}{a}
\newcommand{\generalrvOneAlph}{\mathcal{A}}
\newcommand{\generalrvTwo}{B}
\newcommand{\generalrvTwoValue}{b}
\newcommand{\generalrvTwoAlph}{\mathcal{B}}
\newcommand{\generalrvThree}{C}
\newcommand{\generalrvThreeValue}{c}
\newcommand{\generalrvThreeAlph}{\mathcal{C}}
\newcommand{\generalpdistOne}{P}
\newcommand{\generalpmfOne}{p}
\newcommand{\generalpdistTwo}{Q}
\newcommand{\generalpmfTwo}{q}
\newcommand{\generalSecondMoment}{\sigma}
\newcommand{\generalThirdMoment}{\rho}
\newcommand{\generalcdf}{F}
\newcommand{\generalintegrand}{x}
\newcommand{\generalSummationIndex}{k}
\newcommand{\generalSummationBound}{n}
\newcommand{\coverNumber}{\chi}
\newcommand{\codebookRate}{R}
\newcommand{\codebookRateOne}{R_1}
\newcommand{\codebookRateTwo}{R_2}
\newcommand{\channelInOne}{X}
\newcommand{\channelInOneAlph}{\mathcal{X}}
\newcommand{\channelInOneAlphElement}{x}
\newcommand{\channelInTwo}{Y}
\newcommand{\channelInTwoAlph}{\mathcal{Y}}
\newcommand{\channelInTwoAlphElement}{y}
\newcommand{\channelOut}{Z}
\newcommand{\channelOutAlph}{\mathcal{Z}}
\newcommand{\channelOutAlphElement}{z}
\newcommand{\channel}{\mathcal{W}}
\newcommand{\alphSubset}{A}
\newcommand{\codebook}{\mathcal{C}}
\newcommand{\codebookOne}{\mathcal{C}_1}
\newcommand{\codebookTwo}{\mathcal{C}_2}
\newcommand{\codebookOneWord}[1]{C_1(#1)}
\newcommand{\codebookTwoWord}[1]{C_2(#1)}
\newcommand{\codebookWord}[1]{C(#1)}
\newcommand{\codebookSet}{\mathbb{C}}
\newcommand{\codewordIndex}{m}
\newcommand{\codebookBlocklength}{n}
\newcommand{\blockIndex}{k}
\newcommand{\txIndex}{k}
\newcommand{\mutualInformation}[2]{I(#1;#2)}
\newcommand{\mutualInformationConditional}[3]{I(#1;#2|#3)}

\newcommand{\finalconstOne}{\gamma_1}
\newcommand{\finalconstTwo}{\gamma_2}
\newcommand{\totalvariation}[1]{\lVert #1 \rVert_\mathrm{TV}}
\newcommand{\absolute}[1]{\left\lvert #1 \right\rvert}
\newcommand{\positive}[1]{\left[ #1 \right]^+}
\newcommand{\renyiParam}{\alpha}
\newcommand{\proofconstantOne}{{\beta}}

\newcommand{\informationDensity}[2]{i({#1};{#2})}
\newcommand{\informationDensityConditional}[3]{i({#1};{#2} | {#3})}
\newcommand{\renyidiv}[3]{D_{#1}\left({#2} || {#3}\right)}
\newcommand{\Expectation}{\mathbb{E}}
\newcommand{\Probability}{\mathbb{P}}
\newcommand{\indicator}[1]{1_{#1}}
\newcommand{\cardinality}[1]{\lvert #1 \rvert}
\newcommand{\typicalityParam}{\varepsilon}
\newcommand{\typicalSetIndex}[3]{\mathcal{T}_{#3,#1}^{#2}}
\newcommand{\typicalSet}[2]{\mathcal{T}_{#1}^{#2}}
\newcommand{\lemmaexpectation}{\mu}
\newcommand{\channelDispersion}[1]{V_{#1}}
\newcommand{\channelThirdMoment}[1]{\rho_{#1}}
\newcommand{\normalcdfComplement}{\mathcal{Q}}
\newcommand{\normalcdf}{\Phi}
\newcommand{\normalcdfComplementInverse}{\mathcal{Q}^{-1}}
\newcommand{\secondOrderParamC}{c}
\newcommand{\secondOrderParamD}{d}
\newcommand{\secondOrderAtypicalProbability}[1]{\mu_{#1}}
\newcommand{\totvarAtypicalOne}{P_{\mathrm{atyp}, 1}}
\newcommand{\totvarAtypicalTwo}{P_{\mathrm{atyp}, 2}}
\newcommand{\totvarTypical}[1]{P_{\mathrm{typ}}({#1})}
\newcommand{\totvarTypicalOne}[2]{P_{\mathrm{typ}, 1}({#1},{#2})}

\newtheorem{theorem}{Theorem}

\newtheorem{lemma}{Lemma}
\newtheorem{cor}{Corollary}

\newtheorem{remark}{Remark}

\title{MAC Resolvability: First And Second Order Results}
\author{
Matthias Frey, Igor Bjelaković and Sławomir Stańczak
\\
Technische Universität Berlin
\thanks{
The work was supported by the German Research Foundation (DFG) under grant 
STA864/7-1 and by the German Federal Ministry of Education and Research under 
grant 16KIS0605.
}
}
\DefCurrentAudience{arxiv}
\allowdisplaybreaks

\begin{document}

\maketitle

\begin{abstract}
Building upon previous work on the relation between secrecy and
  channel resolvability, we revisit a secrecy proof for the
  multiple-access channel (MAC) from the perspective of resolvability. We
  then refine the approach in order to obtain some novel results on
  the second-order achievable rates.
\end{abstract}

\section{Introduction}
\label{section:Introduction}

\subsection{Motivation}
With an increasing number of users and things being connected to each
other, not only the overall amount of communication increases, but
also the amount of private and personal information being
transferred. This information needs to be protected from various
attacks. For some potential applications, like e.g. emerging e-health
technologies where sensitive medical data is transmitted using a Body
Area Network, the problem of providing secrecy guarantees is a key
issue. As discovered by Csiszár~\cite{CsiszarSecrecy} and later more
explicitly by Bloch and Laneman~\cite{BlochStrongSecrecy} and investigated by Yassaee and Aref \cite{YassaeeMACWiretap} for the multiple-access case, the concept
of channel resolvability can be applied to provide such guarantees; it
can further be of use as a means of exploiting channel noise in order
to convey randomness to a receiver, where the observed distribution
can be accurately controlled at the transmitter. In this paper, we
explore channel resolvability in a multiple-access setting in which
there is no communication between the transmitters, yet they can
control the distribution observed at the receiver in a non-cooperative
manner.

\subsection{Literature}
\label{section:literature}

To the best of our knowledge, the concept of approximating a desired output distribution over a
communication channel using as little randomness as possible at the
transmitter was first introduced by
Wyner~\cite{WynerCommonInformation}, who used normalized
Kullback-Leibler divergence to measure how close the actual and the
desired output distribution are. The term \emph{channel resolvability}
for a similar concept was introduced by Han and
Verdú~\cite{HanApproximation}, who however used variational distance
as a metric. In particular, they showed the existence of a codebook
that achieves an arbitrarily small variational distance by studying
the expected variational distance of a random codebook.

Resolvability for MACs has been explored by Steinberg~\cite{SteinbergResolvability} and later by Oohama~\cite{OohamaConverse}. Explicit low-complexity codebooks for the special case of symmetric MACs have been proposed by Chou, Bloch and Kliewer~\cite{ChouLowComplexity}.

A stronger result stating that the probability of drawing an
unsuitable random codebook is doubly exponentially small is due to
Cuff~\cite{CuffSoftCovering}. Related results were proposed before by
Csiszár~\cite{CsiszarSecrecy} and by
Devetak~\cite{DevetakPrivateCapacity} for the quantum setting, who
based his work on the non-commutative Chernoff
bound~\cite{AhlswedeIdentification}. Further secrecy results based on
or related to the concept of channel resolvability are due to
Hayashi~\cite{HayashiResolvability}, Bloch and Laneman
\cite{BlochStrongSecrecy}, Hou and Kramer~\cite{HouEffectiveSecrecy},
and Wiese and Boche~\cite{WieseWiretap}, who applied Devetak's
approach to a multiple-access setting.

Cuff~\cite{CuffSoftCovering} also gave a result on the second-order
rate; a related result was proposed by Watanabe and
Hayashi~\cite{WantanabeSecondOrder}.

\subsection{Overview and Outline}
In this work, we revisit the proof in~\cite{WieseWiretap}, while focusing on channel resolvability. We use a slightly different technique as in~\cite{CuffSoftCovering}, which we extend to the multiple-access case to provide an explicit statement and a more intuitive proof for a result only implicitly contained in~\cite{WieseWiretap}, and extend it by providing a second-order result.

In the following section, we state definitions and prior results that we will be using in our proofs in Section~\ref{section:main}.

\section{Notation, Definitions and Prerequisites}
\label{section:preliminaries}
The operations $\log$ and $\exp$ use Euler's number as a basis, and all information quantities are given in nats. $\positive{\cdot}$ denotes the maximum of its argument and $0$.

A \emph{channel}
$\channel = (\channelInOneAlph, \channelInTwoAlph, \channelOutAlph, \channelpmf_{\channelOut | \channelInOne, \channelInTwo})$
is given by finite input alphabets $\channelInOneAlph$ and $\channelInTwoAlph$, a finite output alphabet $\channelOutAlph$ and a collection of probability mass functions $\channelpmf_{\channelOut | \channelInOne, \channelInTwo}$ on $\channelOutAlph$ for each pair $(\channelInOneAlphElement,\channelInTwoAlphElement) \in \channelInOneAlph \times \channelInTwoAlph$. The random variables $\channelInOne$, $\channelInTwo$ and $\channelOut$ represent the two channel inputs and the channel output, respectively. \emph{Input distributions} for the channel are probability mass functions on $\channelInOneAlph$ and $\channelInTwoAlph$ denoted by $\channelpmf_{\channelInOne}$ and $\channelpmf_{\channelInTwo}$, respectively. We define an \emph{induced joint distribution} $\channelpmf_{\channelInOne, \channelInTwo, \channelOut}$ on $\channelInOneAlph \times \channelInTwoAlph \times \channelOutAlph$ by
$\channelpmf_{\channelInOne, \channelInTwo, \channelOut}(\channelInOneAlphElement,\channelInTwoAlphElement,\channelOutAlphElement) :=  \channelpmf_{\channelInOne}(\channelInOneAlphElement) \channelpmf_{\channelInTwo}(\channelInTwoAlphElement) \channelpmf_{\channelOut | \channelInOne, \channelInTwo}(\channelOutAlphElement | \channelInOneAlphElement,\channelInTwoAlphElement)$
and the \emph{output distribution}
$
\channelpmf_\channelOut(\channelOutAlphElement)
:=
\sum_{\channelInOneAlphElement \in \channelInOneAlph}
\sum_{\channelInTwoAlphElement \in \channelInTwoAlph}
  \channelpmf_{\channelInOne, \channelInTwo, \channelOut}(\channelInOneAlphElement, \channelInTwoAlphElement, \channelOutAlphElement)
$
is the marginal distribution of $\channelOut$.

By a pair of \emph{codebooks} of block length $\codebookBlocklength$ and rates $\codebookRateOne$ and $\codebookRateTwo$, we mean finite sequences
$\codebookOne = (\codebookOneWord{\codewordIndex})_{\codewordIndex = 1}^{\exp(\codebookBlocklength\codebookRateOne)}$
and
$\codebookTwo = (\codebookTwoWord{\codewordIndex})_{\codewordIndex = 1}^{\exp(\codebookBlocklength\codebookRateTwo)}$,
where the \emph{codewords} $\codebookOneWord{\codewordIndex} \in \channelInOneAlph^\codebookBlocklength$ and $\codebookTwoWord{\codewordIndex} \in \channelInTwoAlph^\codebookBlocklength$ are finite sequences of elements of the input alphabets. We define a probability distribution $\Probability_{\codebookOne, \codebookTwo}$ on these codebooks as i.i.d. drawings in each component of each codeword according to $\channelpmf_\channelInOne$ and $\channelpmf_\channelInTwo$, respectively. Accordingly, we define the \emph{output distribution induced by $\codebookOne$ and $\codebookTwo$} on $\channelOutAlph^\codebookBlocklength$ by
\begin{multline*}
\codebookpmf_{\channelOut^\codebookBlocklength | \codebookOne, \codebookTwo}(\channelOutAlphElement^\codebookBlocklength) :=
  \exp(-\codebookBlocklength(\codebookRateOne+\codebookRateTwo))
  \\ \cdot
  \sum\limits_{\codewordIndex_1=1}^{\exp(\codebookBlocklength\codebookRateOne)}
  \sum\limits_{\codewordIndex_2=1}^{\exp(\codebookBlocklength\codebookRateTwo)}
      \channelpmf_{\channelOut^\codebookBlocklength | \channelInOne^\codebookBlocklength, \channelInTwo^\codebookBlocklength}(\channelOutAlphElement^\codebookBlocklength | \codebookOneWord{\codewordIndex_1}, \codebookTwoWord{\codewordIndex_2}).
\end{multline*}

Given probability distributions $\generalpdistOne$ and $\generalpdistTwo$ on a finite set $\generalrvOneAlph$ with mass functions $\generalpmfOne$ and $\generalpmfTwo$, respectively, and positive $\renyiParam \neq 1$, the \emph{Rényi divergence of order $\renyiParam$ of $\generalpdistOne$ from $\generalpdistTwo$} is defined as
\[
\renyidiv{\renyiParam}{\generalpdistOne}{\generalpdistTwo}
:=
\frac{1}{\renyiParam-1}
\log
\sum\limits_{\generalrvOneValue \in \generalrvOneAlph}
  \generalpmfOne(\generalrvOneValue)^\renyiParam
  \generalpmfTwo(\generalrvOneValue)^{1-\renyiParam}.
\]
Furthermore, we define the \emph{variational distance} between $\generalpdistOne$ and $\generalpdistTwo$ (or between their mass functions) as
\[
\totalvariation{\generalpmfOne - \generalpmfTwo}
:=
\frac{1}{2} \sum\limits_{\generalrvOneValue \in \generalrvOneAlph} \absolute{\generalpmfOne(\generalrvOneValue) - \generalpmfTwo(\generalrvOneValue)}
=
\sum\limits_{\generalrvOneValue \in \generalrvOneAlph} \positive{\generalpmfOne(\generalrvOneValue) - \generalpmfTwo(\generalrvOneValue)}.
\]

Given random variables $\generalrvOne$, $\generalrvTwo$ and $\generalrvThree$ distributed according to $\generalpmf_{\generalrvOne, \generalrvTwo, \generalrvThree}$, we define the \emph{(conditional) information density} as
\[
\informationDensity{\generalrvOneValue}{\generalrvTwoValue} := \log \frac{\generalpmf_{\generalrvTwo | \generalrvOne}(\generalrvTwoValue | \generalrvOneValue)}{\generalpmf_{\generalrvTwo}(\generalrvTwoValue)}
,~~
\informationDensityConditional{\generalrvOneValue}{\generalrvTwoValue}{\generalrvThreeValue} := \log \frac{\generalpmf_{\generalrvTwo | \generalrvOne, \generalrvThree}(\generalrvTwoValue | \generalrvOneValue, \generalrvThreeValue)}{\generalpmf_{\generalrvTwo | \generalrvThree}(\generalrvTwoValue | \generalrvThreeValue)}.
\]
The (conditional) mutual information is the expected value of the (conditional) information density.

The following inequality was introduced in~\cite{Berry} and~\cite{Esseen}; we use a refinement here which follows e.g. from~\cite{BeekBerryEsseen}.
\begin{theorem}[Berry-Esseen Inequality]
\label{theorem:berry-esseen}
Given a sequence $(\generalrvOne_\generalSummationIndex)_{\generalSummationIndex=1}^{\generalSummationBound}$ of i.i.d. copies of a random variable $\generalrvOne$ on the reals with $\Expectation \generalrvOne = 0$ and finite $\Expectation \generalrvOne^2 = \generalSecondMoment^2$ and $\Expectation \absolute{\generalrvOne}^3 = \generalThirdMoment$, define $\bar{\generalrvOne} := (\generalrvOne_1 + \dots + \generalrvOne_\generalSummationBound)/\generalSummationBound$. Then the cumulative distribution functions $\generalcdf(\generalrvOneValue) := \Probability(\bar{\generalrvOne}\sqrt{\generalSummationBound}/\generalSecondMoment \leq \generalrvOneValue)$ of $\bar{\generalrvOne}\sqrt{\generalSummationBound}/\generalSecondMoment$ and $\normalcdf(\generalrvOneValue) := \int_{-\infty}^\generalrvOneValue 1/(2\pi) \exp(-\generalintegrand^2/2)  d \generalintegrand$ of the standard normal distribution satisfy for all real numbers $\generalrvOneValue$
\[
\absolute{\generalcdf(\generalrvOneValue) - \normalcdf(\generalrvOneValue)}
\leq
\frac{\generalThirdMoment}
     {\generalSecondMoment^3 \sqrt{\generalSummationBound}}.
\]
\end{theorem}
\begin{shownto}{arxiv}
We further use variations of the concentration bounds introduced in~\cite{HoeffdingInequalities}.
\begin{theorem}[Chernoff-Hoeffding Bound]
\label{theorem:hoeffding}
Suppose $\generalrvOne = \sum_{\generalSummationIndex=1}^{\generalSummationBound} \generalrvOne_\generalSummationIndex$, where the random variables in the sequence $(\generalrvOne_\generalSummationIndex)_{\generalSummationIndex=1}^\generalSummationBound$ are independently distributed with values in $[0,1]$ and $\Expectation \generalrvOne \leq \lemmaexpectation$. Then for $0 < \lemmaconst < 1$,
\[
\Probability(\generalrvOne > \lemmaexpectation(1+\lemmaconst)) \leq \exp\left(-\frac{\lemmaconst^2}{3} \lemmaexpectation \right).
\]
\end{theorem}
This version can e.g. be found in~\cite[Ex. 1.1]{ConcentrationTextbook}. We will also be using an extension of the Chernoff-Hoeffding bound for dependent variables due to Janson~\cite[Theorem 2.1]{JansonLargeDeviations}, of which we state only a specialized instance that is used in this paper.
\begin{theorem}[Janson~\cite{JansonLargeDeviations}]
\label{theorem:janson}
Suppose $\generalrvOne = \sum_{\generalSummationIndex=1}^{\generalSummationBound} \generalrvOne_\generalSummationIndex$, where the random variables in the sequence $(\generalrvOne_\generalSummationIndex)_{\generalSummationIndex=1}^\generalSummationBound$ take values in $[0,1]$ and can be partitioned into $\coverNumber \geq 1$ sets such that the random variables in each set are independently distributed. Then, for $\lemmaconst > 0$,
\[
\Probability(\generalrvOne \geq \Expectation \generalrvOne + \lemmaconst)
\leq
\exp\left(
  -2 \frac{\lemmaconst^2}
          {\coverNumber \cdot \generalSummationBound}
\right).
\]
\end{theorem}
\end{shownto}

\section{Main Results}
\label{section:main}
\begin{theorem}
\label{theorem:soft-covering-two-transmitters}
Suppose
$\channel = (\channelInOneAlph, \channelInTwoAlph, \channelOutAlph, \channelpmf_{\channelOut | \channelInOne, \channelInTwo})$
is a channel, $\channelpmf_\channelInOne$ and $\channelpmf_\channelInTwo$ are input distributions, $\codebookRateOne > \mutualInformationConditional{\channelInOne}{\channelOut}{\channelInTwo}$ and $\codebookRateTwo > \mutualInformation{\channelInTwo}{\channelOut}$.
Then there exist $\finalconstOne, \finalconstTwo > 0$ such that for large enough block length $\codebookBlocklength$, the codebook distributions of block length $\codebookBlocklength$ and rates $\codebookRateOne$ and $\codebookRateTwo$ satisfy
\begin{multline}
\label{theorem:soft-covering-two-transmitters-probability-statement}
\Probability_{\codebookOne, \codebookTwo} \left(
  \totalvariation{
    \codebookpmf_{\channelOut^\codebookBlocklength | \codebookOne, \codebookTwo} - \channelpmf_{\channelOut^\codebookBlocklength}
  }
  >
  \exp(-\finalconstOne\codebookBlocklength)
\right)
\\
\leq
\exp\left(-\exp\left(\finalconstTwo\codebookBlocklength\right)\right).
\end{multline}
\end{theorem}
Observing that this theorem can be applied with the roles of $\channelInOne$ and $\channelInTwo$ reversed and that time sharing is possible, we obtain the following corollary.
\begin{cor}
Theorem~\ref{theorem:soft-covering-two-transmitters} holds for all interior points in the convex closure of
\begin{align*}
\{
  (\codebookRateOne, \codebookRateTwo)
  :
  &(
    \codebookRateOne \geq \mutualInformationConditional{\channelInOne}{\channelOut}{\channelInTwo}
    \wedge
    \codebookRateTwo \geq \mutualInformation{\channelInTwo}{\channelOut}
  )
  \\
  &\vee
  (
    \codebookRateOne \geq \mutualInformation{\channelInOne}{\channelOut}
    \wedge
    \codebookRateTwo \geq \mutualInformationConditional{\channelInTwo}{\channelOut}{\channelInOne}
  )
\}.
\end{align*}
\end{cor}

\begin{theorem}
\label{theorem:soft-covering-two-transmitters-second-order}
Given a channel
$\channel = (\channelInOneAlph, \channelInTwoAlph, \channelOutAlph, \channelpmf_{\channelOut | \channelInOne, \channelInTwo})$,
input distributions $\channelpmf_\channelInOne$ and $\channelpmf_\channelInTwo$, $\typicalityParam \in (0,1)$, let the central second and absolute third moment of $\informationDensityConditional{\channelInOne}{\channelOut}{\channelInTwo}$ be $\channelDispersion{1}$ and $\channelThirdMoment{1}$, respectively; analogously, we use $\channelDispersion{2}$ and $\channelThirdMoment{2}$ to denote the central second and absolute third moment of $\informationDensity{\channelInTwo}{\channelOut}$. Suppose the rates $\codebookRateOne, \codebookRateTwo$ depend on $\codebookBlocklength$ in the following way:
\begin{alignat}{3}
\label{theorem:soft-covering-two-transmitters-second-order-rate-one}
\codebookRateOne
&=
\mutualInformationConditional{\channelInOne}{\channelOut}{\channelInTwo}&
&+
\sqrt{\frac{\channelDispersion{1}}{\codebookBlocklength}} \normalcdfComplementInverse(\typicalityParam)&
&+
\secondOrderParamC\frac{\log \codebookBlocklength}
                       {\codebookBlocklength} \\
\label{theorem:soft-covering-two-transmitters-second-order-rate-two}
\codebookRateTwo
&=
\mutualInformation{\channelInTwo}{\channelOut}&
&+
\sqrt{\frac{\channelDispersion{2}}{\codebookBlocklength}} \normalcdfComplementInverse(\typicalityParam)&
&+
\secondOrderParamC\frac{\log \codebookBlocklength}
                       {\codebookBlocklength},
\end{alignat}
where $\normalcdfComplement := 1 - \normalcdf$ with $\normalcdf$ as defined in the statement of Theorem~\ref{theorem:berry-esseen}, and $\secondOrderParamC>1$. Then, for any $\secondOrderParamD \in (0, \secondOrderParamC-1)$, we have
\begin{multline*}
\begin{aligned}
  \Probability_{\codebookOne, \codebookTwo}
  \left( \vphantom{\frac{1}{\sqrt{\codebookBlocklength}}}
  \right.
    &\totalvariation{\codebookpmf_{\channelOut^\codebookBlocklength | \codebookOne, \codebookTwo} - \channelpmf_{\channelOut^\codebookBlocklength}}
    >
    \\ &\left.
    (\secondOrderAtypicalProbability{1} + \secondOrderAtypicalProbability{2})
    \left(1+\frac{1}{\sqrt{\codebookBlocklength}}\right)
    +
    \frac{3}{\sqrt{\codebookBlocklength}}
  \right)
\end{aligned}
\\
\leq
2\exp\left(
  -\frac{2\min(\secondOrderAtypicalProbability{1}^2,\secondOrderAtypicalProbability{2}^2)}
        {\codebookBlocklength}
  \exp(\codebookBlocklength \min(\codebookRateOne,\codebookRateTwo))
\right) \\
+
2\exp\left(
  \codebookBlocklength(\log \cardinality{\channelOutAlph} + \log \cardinality{\channelInTwoAlph})
  -\frac{1}{3}
  \codebookBlocklength^{\secondOrderParamC - \secondOrderParamD - 1}
\right),
\end{multline*}
where for both $\txIndex=1$ and $\txIndex=2$,
\begin{align*}
\secondOrderAtypicalProbability{\txIndex}
:=
\normalcdfComplement\left(
  \normalcdfComplementInverse(\typicalityParam)
  +
  \frac{\secondOrderParamD \log \codebookBlocklength}
       {\sqrt{\codebookBlocklength\channelDispersion{\txIndex}}}
\right)
+
\frac{\channelThirdMoment{\txIndex}}
     {\channelDispersion{\txIndex}^{\frac{3}{2}} \sqrt{\codebookBlocklength}}
\end{align*}
tends to $\typicalityParam$ for $\codebookBlocklength \rightarrow \infty$.
\end{theorem}
Again, observing that this theorem can be applied with the roles of $\channelInOne$ and $\channelInTwo$ reversed, we have
\begin{cor}
\label{cor:soft-covering-two-transmitters-second-order}
Theorem~\ref{theorem:soft-covering-two-transmitters-second-order} holds with (\ref{theorem:soft-covering-two-transmitters-second-order-rate-one}) and (\ref{theorem:soft-covering-two-transmitters-second-order-rate-two}) replaced by
\begin{align*}
\codebookRateOne
&=
\mutualInformation{\channelInOne}{\channelOut}
+
\sqrt{\frac{\channelDispersion{1}}{\codebookBlocklength}} \normalcdfComplementInverse(\typicalityParam)
+
\secondOrderParamC\frac{\log \codebookBlocklength}
                       {\codebookBlocklength} \\
\codebookRateTwo
&=
\mutualInformationConditional{\channelInTwo}{\channelOut}{\channelInOne}
+
\sqrt{\frac{\channelDispersion{2}}{\codebookBlocklength}} \normalcdfComplementInverse(\typicalityParam)
+
\secondOrderParamC\frac{\log \codebookBlocklength}
                       {\codebookBlocklength}
\end{align*}
and $\channelDispersion{1}$, $\channelThirdMoment{1}$, $\channelDispersion{2}$ and $\channelThirdMoment{2}$ redefined to be the second and third moments of $\informationDensity{\channelInOne}{\channelOut}$ and $\informationDensityConditional{\channelInTwo}{\channelOut}{\channelInOne}$, respectively.
\end{cor}
\begin{remark}
The question of how the achievable second-order rates behave near the line connecting the two corner points should be a subject of further research.
\end{remark}
In the proofs of these theorems, we consider two types of typical sets:
\begin{align*}
  \typicalSetIndex{\typicalityParam}{\codebookBlocklength}{1}
  &:=
  \{
    (\channelInOneAlphElement^\codebookBlocklength, \channelInTwoAlphElement^\codebookBlocklength,\channelOutAlphElement^\codebookBlocklength)
    :
    \informationDensityConditional{\channelInOneAlphElement^\codebookBlocklength}{\channelOutAlphElement^\codebookBlocklength}{\channelInTwoAlphElement^\codebookBlocklength}
    \leq
    \codebookBlocklength(\mutualInformationConditional{\channelInOne}{\channelOut}{\channelInTwo}+\typicalityParam)
  \}
\\
   \typicalSetIndex{\typicalityParam}{\codebookBlocklength}{2}
   &:=
   \{
     (\channelInTwoAlphElement^\codebookBlocklength,\channelOutAlphElement^\codebookBlocklength)
     :
     \informationDensity{\channelInTwoAlphElement^\codebookBlocklength}{\channelOutAlphElement^\codebookBlocklength}
    \leq
    \codebookBlocklength(\mutualInformation{\channelInTwo}{\channelOut}+\typicalityParam)
  \}.
\end{align*}
We split the variational distance in atypical and typical parts as follows, where $\totvarAtypicalOne$, $\totvarAtypicalTwo$ and $\totvarTypical{\channelOutAlphElement^\codebookBlocklength}$ are defined by~(\ref{def:soft-covering-atypical-term-one}), (\ref{def:soft-covering-atypical-term-two}) and (\ref{def:soft-covering-typical-term}) shown on the next page.

\begin{figure*}
\normalsize
\begin{align}
\label{def:soft-covering-atypical-term-one}
\totvarAtypicalOne
&:=
\sum\limits_{\channelOutAlphElement^\codebookBlocklength \in \channelOutAlph^\codebookBlocklength}
\exp(-\codebookBlocklength(\codebookRateOne+\codebookRateTwo))
\sum\limits_{\codewordIndex_1=1}^{\exp(\codebookBlocklength\codebookRateOne)}
\sum\limits_{\codewordIndex_2=1}^{\exp(\codebookBlocklength\codebookRateTwo)}
    \channelpmf_{\channelOut^\codebookBlocklength | \channelInOne^\codebookBlocklength, \channelInTwo^\codebookBlocklength}(\channelOutAlphElement^\codebookBlocklength | \codebookOneWord{\codewordIndex_1}, \codebookTwoWord{\codewordIndex_2})
    \indicator{(\codebookOneWord{\codewordIndex_1}, \codebookTwoWord{\codewordIndex_2}, \channelOutAlphElement^\codebookBlocklength) \notin \typicalSetIndex{\typicalityParam}{\codebookBlocklength}{1}}
\\
\label{def:soft-covering-atypical-term-two}
\totvarAtypicalTwo
&:=
\sum\limits_{\channelOutAlphElement^\codebookBlocklength \in \channelOutAlph^\codebookBlocklength}
\exp(-\codebookBlocklength(\codebookRateOne+\codebookRateTwo))
\sum\limits_{\codewordIndex_1=1}^{\exp(\codebookBlocklength\codebookRateOne)}
\sum\limits_{\codewordIndex_2=1}^{\exp(\codebookBlocklength\codebookRateTwo)}
    \channelpmf_{\channelOut^\codebookBlocklength | \channelInOne^\codebookBlocklength, \channelInTwo^\codebookBlocklength}(\channelOutAlphElement^\codebookBlocklength | \codebookOneWord{\codewordIndex_1}, \codebookTwoWord{\codewordIndex_2})
    \indicator{(\codebookTwoWord{\codewordIndex_2}, \channelOutAlphElement^\codebookBlocklength) \notin \typicalSetIndex{\typicalityParam}{\codebookBlocklength}{2}}
\\
\label{def:soft-covering-typical-term}
\totvarTypical{\channelOutAlphElement^\codebookBlocklength}
&:=
    \sum\limits_{\codewordIndex_1=1}^{\exp(\codebookBlocklength\codebookRateOne)}
    \sum\limits_{\codewordIndex_2=1}^{\exp(\codebookBlocklength\codebookRateTwo)}
        \exp(-\codebookBlocklength(\codebookRateOne+\codebookRateTwo))
        \frac{\channelpmf_{\channelOut^\codebookBlocklength | \channelInOne^\codebookBlocklength, \channelInTwo^\codebookBlocklength}(\channelOutAlphElement^\codebookBlocklength | \codebookOneWord{\codewordIndex_1}, \codebookTwoWord{\codewordIndex_2})}
             {\channelpmf_{\channelOut^\codebookBlocklength}(\channelOutAlphElement^\codebookBlocklength)}
        \indicator{(\codebookTwoWord{\codewordIndex_2}, \channelOutAlphElement^\codebookBlocklength) \in \typicalSetIndex{\typicalityParam}{\codebookBlocklength}{2}}
        \indicator{(\codebookOneWord{\codewordIndex_1}, \codebookTwoWord{\codewordIndex_2}, \channelOutAlphElement^\codebookBlocklength) \in \typicalSetIndex{\typicalityParam}{\codebookBlocklength}{1}}
\end{align}
\hrulefill
\end{figure*}

\begin{align}
\notag
&\totalvariation{ \codebookpmf_{\channelOut^\codebookBlocklength | \codebookOne, \codebookTwo} - \channelpmf_{\channelOut^\codebookBlocklength}}
\\
\notag
=
&\sum\limits_{\channelOutAlphElement^\codebookBlocklength \in \channelOutAlph^\codebookBlocklength}
  \channelpmf_{\channelOut^\codebookBlocklength}(\channelOutAlphElement^\codebookBlocklength)
  \positive{\frac{\codebookpmf_{\channelOut^\codebookBlocklength | \codebookOne, \codebookTwo}(\channelOutAlphElement^\codebookBlocklength)}
                 {\channelpmf_{\channelOut^\codebookBlocklength}(\channelOutAlphElement^\codebookBlocklength)}
  - 1
  }
\\
\label{proof:soft-covering-two-transmitters-typical-split}
\leq
&\totvarAtypicalOne + \totvarAtypicalTwo
+
\sum\limits_{\channelOutAlphElement^\codebookBlocklength \in \channelOutAlph^\codebookBlocklength}
\channelpmf_{\channelOut^\codebookBlocklength}(\channelOutAlphElement^\codebookBlocklength)
\positive{\totvarTypical{\channelOutAlphElement^\codebookBlocklength} - 1}.
\end{align}

\begin{remark}
The denominator of the fraction is almost surely not equal to $0$ as long as the numerator is not equal to $0$. We implicitly let the summation range only over the support of the denominator, as we do in all further summations.
\end{remark}

So the theorems can be proven by considering typical and atypical terms separately.
\showto{arxiv}{But first, we prove two lemmas to help us to bound the typical and the atypical terms.}
\showto{conference}{But first, we state two lemmas to help us bound these terms. The proofs use Chernoff-Hoeffding concentration bounds introduced in~\cite{HoeffdingInequalities} and an extension thereof for dependent variables by Janson~\cite[Theorem 2.1]{JansonLargeDeviations}. We omit the proofs here due to lack of space, however, they can be found in the extended version of this paper~\cite{arxivVersion}.}

\begin{lemma}[Bound for typical terms]
\label{lemma:soft-covering-two-transmitters-typical}
Given a block length $\codebookBlocklength$, $\typicalityParam > 0$, $0 < \lemmaconst < 1$, random variables $\generalrvOne$, $\generalrvTwo$ and $\generalrvThree$ on finite alphabets $\generalrvOneAlph$, $\generalrvTwoAlph$ and $\generalrvThreeAlph$ respectively with joint probability mass function $\generalpmf_{\generalrvOne, \generalrvTwo, \generalrvThree}$, a rate $\codebookRate$ and a codebook
$\codebook = (\codebookWord{\codewordIndex})_{\codewordIndex=1}^{\exp(\codebookBlocklength\codebookRate)}$ with each component of each codeword drawn i.i.d. according to $\generalpmf_\generalrvOne$, for any $\generalrvTwoValue^\codebookBlocklength \in \generalrvTwoAlph^\codebookBlocklength$ and $\generalrvThreeValue^\codebookBlocklength \in \generalrvThreeAlph^\codebookBlocklength$, we have
\begin{multline*}
\showto{arxiv}{\check{\Probability} :=}
\Probability_{\codebook}\left(
  \sum\limits_{\codewordIndex=1}^{\exp(\codebookBlocklength\codebookRate)}
  \exp(-\codebookBlocklength\codebookRate)
  \frac{\generalpmf_{\generalrvThree^\codebookBlocklength | \generalrvOne^\codebookBlocklength, \generalrvTwo^\codebookBlocklength}(\generalrvThreeValue^\codebookBlocklength | \codebookWord{\codewordIndex}, \generalrvTwoValue^\codebookBlocklength)}
       {\generalpmf_{\generalrvThree^\codebookBlocklength | \generalrvTwo^\codebookBlocklength}(\generalrvThreeValue^\codebookBlocklength | \generalrvTwoValue^\codebookBlocklength)}
  \right.
  \\
  \left.
  \vphantom{\sum\limits_{\codewordIndex=1}^{\exp(\codebookBlocklength\codebookRate)}}
  \cdot
  \indicator{(\codebookWord{\codewordIndex}, \generalrvTwoValue^\codebookBlocklength, \generalrvThreeValue^\codebookBlocklength) \in \typicalSet{\typicalityParam}{\codebookBlocklength}}
  >
  1 + \lemmaconst
\right) \\
\leq
\exp\left(
  -\frac{\lemmaconst^2}{3} \exp(-\codebookBlocklength (\mutualInformationConditional{\generalrvOne}{\generalrvThree}{\generalrvTwo} + \typicalityParam - \codebookRate))
\right),
\end{multline*}
where the typical set is defined as
\begin{shownto}{arxiv}
\begin{align}
\label{lemma:soft-covering-two-transmitters-typical-def}
\typicalSet{\typicalityParam}{\codebookBlocklength}
:=
\{
  (\generalrvOneValue^\codebookBlocklength, \generalrvTwoValue^\codebookBlocklength, \generalrvThreeValue^\codebookBlocklength)
  :
  \informationDensityConditional{\generalrvOneValue^\codebookBlocklength}{\generalrvThreeValue^\codebookBlocklength}{\generalrvTwoValue^\codebookBlocklength}
  \leq
  \codebookBlocklength(\mutualInformationConditional{\generalrvOne}{\generalrvThree}{\generalrvTwo}+\typicalityParam)
\}.
\end{align}
\end{shownto}
\begin{shownto}{conference}
\begin{align*}
\typicalSet{\typicalityParam}{\codebookBlocklength}
:=
\{
  (\generalrvOneValue^\codebookBlocklength, \generalrvTwoValue^\codebookBlocklength, \generalrvThreeValue^\codebookBlocklength)
  :
  \informationDensityConditional{\generalrvOneValue^\codebookBlocklength}{\generalrvThreeValue^\codebookBlocklength}{\generalrvTwoValue^\codebookBlocklength}
  \leq
  \codebookBlocklength(\mutualInformationConditional{\generalrvOne}{\generalrvThree}{\generalrvTwo}+\typicalityParam)
\}.
\end{align*}
\end{shownto}
\end{lemma}

\begin{shownto}{arxiv}
\begin{proof}
We have
\begin{multline*}
\check{\Probability} =
\Probability_{\codebook}\left(
  \sum\limits_{\codewordIndex=1}^{\exp(\codebookBlocklength\codebookRate)}
  \exp(-\codebookBlocklength (\mutualInformationConditional{\generalrvOne}{\generalrvThree}{\generalrvTwo} + \typicalityParam))
  \right.
  \\
  \cdot
  \frac{\generalpmf_{\generalrvThree^\codebookBlocklength | \generalrvOne^\codebookBlocklength, \generalrvTwo^\codebookBlocklength}(\generalrvThreeValue^\codebookBlocklength | \codebookWord{\codewordIndex}, \generalrvTwoValue^\codebookBlocklength)}
       {\generalpmf_{\generalrvThree^\codebookBlocklength | \generalrvTwo^\codebookBlocklength}(\generalrvThreeValue^\codebookBlocklength | \generalrvTwoValue^\codebookBlocklength)}
  \cdot
  \indicator{(\codebookWord{\codewordIndex}, \generalrvTwoValue^\codebookBlocklength, \generalrvThreeValue^\codebookBlocklength) \in \typicalSet{\typicalityParam}{\codebookBlocklength}}
  \\
  >
  \left. \vphantom{\sum\limits_{\codewordIndex=1}^{\exp(\codebookBlocklength\codebookRate)}}
  \exp(-\codebookBlocklength (\mutualInformationConditional{\generalrvOne}{\generalrvThree}{\generalrvTwo} + \typicalityParam - \codebookRate))
  (1 + \lemmaconst)
  \right).
\end{multline*}
By the definition of $\typicalSet{\typicalityParam}{\codebookBlocklength}$ in~(\ref{lemma:soft-covering-two-transmitters-typical-def}), the summands are at most $1$, and furthermore, the expectation of the sum can be bounded as
\begin{align*}
&
\begin{aligned}
  \Expectation_{\codebook}\left(
    \sum\limits_{\codewordIndex=1}^{\exp(\codebookBlocklength\codebookRate)}
    \right.
    &\exp(-\codebookBlocklength (\mutualInformationConditional{\generalrvOne}{\generalrvThree}{\generalrvTwo} + \typicalityParam))
    \\
    &\cdot
    \frac{\generalpmf_{\generalrvThree^\codebookBlocklength | \generalrvOne^\codebookBlocklength, \generalrvTwo^\codebookBlocklength}(\generalrvThreeValue^\codebookBlocklength | \codebookWord{\codewordIndex}, \generalrvTwoValue^\codebookBlocklength)}
        {\generalpmf_{\generalrvThree^\codebookBlocklength | \generalrvTwo^\codebookBlocklength}(\generalrvThreeValue^\codebookBlocklength | \generalrvTwoValue^\codebookBlocklength)}
    \indicator{(\codebookWord{\codewordIndex}, \generalrvTwoValue^\codebookBlocklength, \generalrvThreeValue^\codebookBlocklength) \in \typicalSet{\typicalityParam}{\codebookBlocklength}}
  \left.
  \vphantom{\sum\limits_{\codewordIndex=1}^{\exp(\codebookBlocklength\codebookRate)}}
  \right)
\end{aligned}
\\
&
\begin{aligned}
  \leq
  \sum\limits_{\codewordIndex=1}^{\exp(\codebookBlocklength\codebookRate)}
  &\exp(-\codebookBlocklength (\mutualInformationConditional{\generalrvOne}{\generalrvThree}{\generalrvTwo} + \typicalityParam))
  \\
  &\cdot
  \Expectation_{\codebook}\left(
    \frac{\generalpmf_{\generalrvThree^\codebookBlocklength | \generalrvOne^\codebookBlocklength, \generalrvTwo^\codebookBlocklength}(\generalrvThreeValue^\codebookBlocklength | \codebookWord{\codewordIndex}, \generalrvTwoValue^\codebookBlocklength)}
        {\generalpmf_{\generalrvThree^\codebookBlocklength | \generalrvTwo^\codebookBlocklength}(\generalrvThreeValue^\codebookBlocklength | \generalrvTwoValue^\codebookBlocklength)}
  \right)
\end{aligned}
\\
&=
\exp(-\codebookBlocklength (\mutualInformationConditional{\generalrvOne}{\generalrvThree}{\generalrvTwo} + \typicalityParam - \codebookRate)).
\end{align*}
Now applying Theorem~\ref{theorem:hoeffding} to the above shows the desired probability statement and completes the proof.
\end{proof}
\end{shownto}

\begin{lemma}[Bound for atypical terms]
\label{lemma:soft-covering-two-transmitters-atypical}
Given a channel
$\channel = (\channelInOneAlph, \channelInTwoAlph, \channelOutAlph, \channelpmf_{\channelOut | \channelInOne, \channelInTwo})$,
input distributions $\channelpmf_\channelInOne$ and $\channelpmf_\channelInTwo$, some set $\alphSubset \subseteq \channelInOneAlph^\codebookBlocklength \times \channelInTwoAlph^\codebookBlocklength \times \channelOutAlph^\codebookBlocklength$, $\lemmaconst > 0$, $\lemmaexpectation \geq \Probability((\channelInOne^\codebookBlocklength, \channelInTwo^\codebookBlocklength, \channelOut^\codebookBlocklength) \in \alphSubset)$ as well as rates $\codebookRateOne$ and $\codebookRateTwo$ and codebooks distributed according to $\Probability_{\codebookOne, \codebookTwo}$ defined in Section~\ref{section:preliminaries}, we have
\begin{multline*}
\showto{arxiv}{\hat{\Probability} :=}
\Probability_{\codebookOne,\codebookTwo}\left(
  \sum\limits_{\channelOutAlphElement^\codebookBlocklength \in \channelOutAlph^\codebookBlocklength}
  \exp(-\codebookBlocklength(\codebookRateOne+\codebookRateTwo))
  \right.
  \\
  \sum\limits_{\codewordIndex_1=1}^{\exp(\codebookBlocklength\codebookRateOne)}
  \sum\limits_{\codewordIndex_2=1}^{\exp(\codebookBlocklength\codebookRateTwo)}
      \channelpmf_{\channelOut^\codebookBlocklength | \channelInOne^\codebookBlocklength, \channelInTwo^\codebookBlocklength}(\channelOutAlphElement^\codebookBlocklength | \codebookOneWord{\codewordIndex_1}, \codebookTwoWord{\codewordIndex_2})
      \\
      \indicator{(\codebookOneWord{\codewordIndex_1}, \codebookTwoWord{\codewordIndex_2}, \channelOutAlphElement^\codebookBlocklength) \in \alphSubset}
  \left. \vphantom{\sum\limits_{\channelOutAlphElement^\codebookBlocklength \in \channelOutAlph^\codebookBlocklength}} >
  \lemmaexpectation(1+\lemmaconst)
\right) \\
\leq
\exp(-2 \lemmaconst^2 \lemmaexpectation^2 \exp(\codebookBlocklength\min(\codebookRateOne,\codebookRateTwo))).
\end{multline*}
\end{lemma}
\begin{shownto}{arxiv}
\begin{proof}
We have
\begin{align*}
&
\begin{aligned}
\hat{\Probability} =
&\Probability_{\codebookOne,\codebookTwo}\left(
  \sum\limits_{\codewordIndex_1=1}^{\exp(\codebookBlocklength\codebookRateOne)}
  \sum\limits_{\codewordIndex_2=1}^{\exp(\codebookBlocklength\codebookRateTwo)}
  \sum\limits_{\channelOutAlphElement^\codebookBlocklength \in \channelOutAlph^\codebookBlocklength}
  \right.
  \\
      &~
      \channelpmf_{\channelOut^\codebookBlocklength | \channelInOne^\codebookBlocklength, \channelInTwo^\codebookBlocklength}(\channelOutAlphElement^\codebookBlocklength | \codebookOneWord{\codewordIndex_1}, \codebookTwoWord{\codewordIndex_2})
      \indicator{(\codebookOneWord{\codewordIndex_1}, \codebookTwoWord{\codewordIndex_2}, \channelOutAlphElement^\codebookBlocklength) \in \alphSubset}
  \\
  &\left. \vphantom{\sum\limits_{\channelOutAlphElement^\codebookBlocklength \in \channelOutAlph^\codebookBlocklength}} >
  \exp(\codebookBlocklength(\codebookRateOne+\codebookRateTwo))
  (
    \lemmaexpectation
    +
    \lemmaexpectation
    \lemmaconst
  )
\right)
\end{aligned}
\\
&
\begin{aligned}
\leq
&\Probability_{\codebookOne,\codebookTwo}\left(
  \sum\limits_{\codewordIndex_1=1}^{\exp(\codebookBlocklength\codebookRateOne)}
  \sum\limits_{\codewordIndex_2=1}^{\exp(\codebookBlocklength\codebookRateTwo)}
  \sum\limits_{\channelOutAlphElement^\codebookBlocklength \in \channelOutAlph^\codebookBlocklength}
  \right.
  \\
      &~
      \channelpmf_{\channelOut^\codebookBlocklength | \channelInOne^\codebookBlocklength, \channelInTwo^\codebookBlocklength}(\channelOutAlphElement^\codebookBlocklength | \codebookOneWord{\codewordIndex_1}, \codebookTwoWord{\codewordIndex_2})
      \indicator{(\codebookOneWord{\codewordIndex_1}, \codebookTwoWord{\codewordIndex_2}, \channelOutAlphElement^\codebookBlocklength) \in \alphSubset}
  \\
  &\left. \vphantom{\sum\limits_{\channelOutAlphElement^\codebookBlocklength \in \channelOutAlph^\codebookBlocklength}} >
  \exp\Big(\codebookBlocklength(\codebookRateOne+\codebookRateTwo)\Big)
  \Big(
    \Probability((\channelInOne^\codebookBlocklength, \channelInTwo^\codebookBlocklength, \channelOut^\codebookBlocklength) \in \alphSubset)
    +
    \lemmaexpectation
    \lemmaconst
  \Big)
\right)
\end{aligned}
\\
&\leq
\exp\left(
  -2\frac{\exp(2\codebookBlocklength(\codebookRateOne+\codebookRateTwo))\lemmaexpectation^2\lemmaconst^2}
         {\exp(\codebookBlocklength\max(\codebookRateOne,\codebookRateTwo)) \exp(\codebookBlocklength(\codebookRateOne + \codebookRateTwo))}
\right)
\\
&=
\exp(-2 \lemmaconst^2 \lemmaexpectation^2 \exp(\codebookBlocklength\min(\codebookRateOne,\codebookRateTwo))),
\end{align*}
where the inequality follows from Theorem~\ref{theorem:janson} by observing that the innermost sum is confined to $[0,1]$, the two outer summations together have $\exp(\codebookBlocklength(\codebookRateOne+\codebookRateTwo)$ summands which can be partitioned into $\exp(\codebookBlocklength(\max(\codebookRateOne,\codebookRateTwo))$ sets with $\exp(\codebookBlocklength\min(\codebookRateOne,\codebookRateTwo))$ independently distributed elements each, and the overall expectation of the term is $\exp(\codebookBlocklength(\codebookRateOne+\codebookRateTwo)\Probability((\channelInOne^\codebookBlocklength, \channelInTwo^\codebookBlocklength, \channelOut^\codebookBlocklength) \in \alphSubset)$.
\end{proof}
\end{shownto}

\begin{proof}[Proof of Theorem~\ref{theorem:soft-covering-two-transmitters}]
In order to bound $\totvarAtypicalOne$, we observe that for any $\renyiParam > 1$, we can bound
$\Probability_{\channelInOne^\codebookBlocklength, \channelInTwo^\codebookBlocklength, \channelOut^\codebookBlocklength}((\channelInOne^\codebookBlocklength, \channelInTwo^\codebookBlocklength, \channelOut^\codebookBlocklength) \notin \typicalSetIndex{\typicalityParam}{\codebookBlocklength}{1})$
as shown in~(\ref{proof:soft-covering-two-transmitters-probability-bound-start}) to~(\ref{proof:soft-covering-two-transmitters-probability-bound-end}) in the appendix,
where the inequality in (\ref{proof:soft-covering-two-transmitters-probability-bound-end}) holds as long as $\proofconstantOne < (\renyiParam-1)(\mutualInformationConditional{\channelInOne}{\channelOut}{\channelInTwo}+\typicalityParam-\renyidiv{\renyiParam}{\Probability_{\channelInOne, \channelInTwo, \channelOut}}{\Probability_{\channelInOne | \channelInTwo}\Probability_{\channelOut | \channelInTwo}\Probability_\channelInTwo})$. We can achieve this for sufficiently small $\proofconstantOne > 0$ as long as $\renyiParam>1$ and $\mutualInformationConditional{\channelInOne}{\channelOut}{\channelInTwo}+\typicalityParam-\renyidiv{\renyiParam}{\Probability_{\channelInOne, \channelInTwo, \channelOut}}{\Probability_{\channelInOne | \channelInTwo}\Probability_{\channelOut | \channelInTwo}\Probability_\channelInTwo} > 0$. In order to choose an $\renyiParam > 1$ such that the latter requirement holds, note that since our alphabets are finite, the Rényi divergence is also finite and thus it is continuous and approaches the Kullback-Leibler divergence for $\renyiParam$ tending to $1$~\cite{RenyiDiv}, which is in this case equal to the mutual information term.

We apply Lemma~\ref{lemma:soft-covering-two-transmitters-atypical} with $\alphSubset = (\channelInOneAlph^\codebookBlocklength \times \channelInTwoAlph^\codebookBlocklength \times \channelOutAlph^\codebookBlocklength) \setminus \typicalSetIndex{\typicalityParam}{\codebookBlocklength}{1}$ and $\lemmaconst = 1$ to obtain
\begin{multline}
\label{proof:soft-covering-two-transmitters-atypical-bound-1}
\Probability_{\codebookOne, \codebookTwo}\left(
  \totvarAtypicalOne
  >
  2\exp(-\codebookBlocklength\proofconstantOne)
\right)
\\
\leq
\exp(
  -2\exp(
    \codebookBlocklength(
      \min(\codebookRateOne,\codebookRateTwo) - 2\proofconstantOne
    )
  )
).
\end{multline}
Proceeding along similar lines of reasoning including another application of Lemma~\ref{lemma:soft-covering-two-transmitters-atypical} with $\alphSubset = \channelInOneAlph^\codebookBlocklength \times ((\channelInTwoAlph^\codebookBlocklength \times \channelOutAlph^\codebookBlocklength) \setminus \typicalSetIndex{\typicalityParam}{\codebookBlocklength}{2})$ and $\lemmaconst=1$, we show that if $\proofconstantOne>0$ is small enough,
\begin{multline}
\label{proof:soft-covering-two-transmitters-atypical-bound-2}
\Probability_{\codebookOne, \codebookTwo}\left(
  \totvarAtypicalTwo
  >
  2\exp(-\codebookBlocklength\proofconstantOne)
\right)
\\
\leq
\exp(
  -2\exp(
    \codebookBlocklength(
      \min(\codebookRateOne,\codebookRateTwo) - 2\proofconstantOne
    )
  )
).
\end{multline}
As for the typical term, we first observe that for any fixed $\channelInTwoAlphElement^\codebookBlocklength$ and $\channelOutAlphElement^\codebookBlocklength$, we can apply Lemma~\ref{lemma:soft-covering-two-transmitters-typical} with $\generalrvOne=\channelInOne$, $\generalrvTwo=\channelInTwo$, $\generalrvThree=\channelOut$ and $\lemmaconst=\exp(-\codebookBlocklength\proofconstantOne)$ to obtain
\begin{multline}
\label{proof:soft-covering-two-transmitters-typical-bound}
\Probability_{\codebookOne}\left(
  \totvarTypicalOne{\channelInTwoAlphElement^\codebookBlocklength}{\channelOutAlphElement^\codebookBlocklength}
  >
  1 + \exp(-\codebookBlocklength\proofconstantOne)
\right)
\\
\leq
\exp\left(
  -\frac{1}{3} \exp(-\codebookBlocklength (\mutualInformationConditional{\channelInOne}{\channelOut}{\channelInTwo} + \typicalityParam + 2\proofconstantOne - \codebookRateOne))
\right),
\end{multline}
where we used
\begin{multline}
\label{def:soft-covering-typical-term-one}
\totvarTypicalOne{\channelInTwoAlphElement^\codebookBlocklength}{\channelOutAlphElement^\codebookBlocklength}
:=
\sum\limits_{\codewordIndex_1=1}^{\exp(\codebookBlocklength\codebookRateOne)}
    \exp(-\codebookBlocklength(\codebookRateOne))
    \\
    \cdot \frac{\channelpmf_{\channelOut^\codebookBlocklength | \channelInOne^\codebookBlocklength, \channelInTwo^\codebookBlocklength}(\channelOutAlphElement^\codebookBlocklength | \codebookOneWord{\codewordIndex_1}, \channelInTwoAlphElement^\codebookBlocklength)}
          {\channelpmf_{\channelOut^\codebookBlocklength | \channelInTwo^\codebookBlocklength}(\channelOutAlphElement^\codebookBlocklength | \channelInTwoAlphElement^\codebookBlocklength)}
    \indicator{(\codebookOneWord{\codewordIndex_1}, \channelInTwoAlphElement^\codebookBlocklength, \channelOutAlphElement^\codebookBlocklength) \in \typicalSetIndex{\typicalityParam}{\codebookBlocklength}{1}}.
\end{multline}
We define a set of codebooks
\begin{align}
\label{def:soft-covering-good-codebooks}
\codebookSet_{\channelOutAlphElement^\codebookBlocklength}
:=
\bigcap\limits_{\channelInTwoAlphElement^\codebookBlocklength \in \channelInTwoAlph^\codebookBlocklength}
  \left\{
    \codebookOne:
    \totvarTypicalOne{\channelInTwoAlphElement^\codebookBlocklength}{\channelOutAlphElement^\codebookBlocklength}
    \leq
    1 + \exp(-\codebookBlocklength\proofconstantOne)
  \right\}
\end{align}
and bound for arbitrary but fixed $\channelOutAlphElement^\codebookBlocklength$
\begin{align*}
\tilde{\Probability}
:=
\Probability_{\codebookOne, \codebookTwo}\left(
  \totvarTypical{\channelOutAlphElement^\codebookBlocklength}
  >
  1 + 3\exp(-\codebookBlocklength\proofconstantOne)
  ~|~
  \codebookOne \in \codebookSet_{\channelOutAlphElement^\codebookBlocklength}
\right)
\end{align*}
in~(\ref{proof:soft-covering-two-transmitters-totalprob1}) to~(\ref{proof:soft-covering-two-transmitters-lemmaapplication2}) in the appendix, where~(\ref{proof:soft-covering-two-transmitters-totalprob1}) follows from the law of total probability,~(\ref{proof:soft-covering-two-transmitters-lemmaapplication1}) is a consequence of the condition $\codebookOne \in \codebookSet_{\channelOutAlphElement^\codebookBlocklength}$,~(\ref{proof:soft-covering-two-transmitters-totalprob2-and-bound}) results from an application of the law of total probability and the assumption that $\codebookBlocklength$ is sufficiently large such that $\exp(-\codebookBlocklength\proofconstantOne) \leq 1$. Finally,~(\ref{proof:soft-covering-two-transmitters-lemmaapplication2}) follows from Lemma~\ref{lemma:soft-covering-two-transmitters-typical} with $\generalrvOne=\channelInTwo$, $\generalrvThree=\channelOut$, $\generalrvTwo$ a deterministic random variable with only one possible realization and $\lemmaconst=\exp(-\codebookBlocklength\proofconstantOne)$.

We can now put everything together as shown in (\ref{proof:soft-covering-two-transmitters-union-bound-start}) to (\ref{proof:soft-covering-two-transmitters-union-bound-substitutions}) in the appendix, where~(\ref{proof:soft-covering-two-transmitters-union-bound-application}) follows from~(\ref{proof:soft-covering-two-transmitters-typical-split}) and the union bound and~(\ref{proof:soft-covering-two-transmitters-union-bound-substitutions}) is a substitution of~(\ref{proof:soft-covering-two-transmitters-atypical-bound-1}), (\ref{proof:soft-covering-two-transmitters-atypical-bound-2}), (\ref{proof:soft-covering-two-transmitters-typical-bound}) and~(\ref{proof:soft-covering-two-transmitters-lemmaapplication2}).

What remains is to choose $\finalconstOne$ and $\finalconstTwo$ such that (\ref{theorem:soft-covering-two-transmitters-probability-statement}) holds. First, we have to choose $\typicalityParam$ and $\proofconstantOne$ small enough such that the terms $\min(\codebookRateOne,\codebookRateTwo)-2\proofconstantOne$, $\codebookRateOne - 2\proofconstantOne - \typicalityParam - \mutualInformationConditional{\channelInOne}{\channelOut}{\channelInTwo}$ and $\codebookRateTwo - 2\proofconstantOne - \typicalityParam - \mutualInformation{\channelInTwo}{\channelOut}$ are all positive. Since there have so far been no constraints on $\proofconstantOne$ and $\typicalityParam$ except that they are positive and sufficiently small, such a choice is possible provided $\codebookRateOne > \mutualInformationConditional{\channelInOne}{\channelOut}{\channelInTwo}$ and $\codebookRateTwo > \mutualInformation{\channelInTwo}{\channelOut}$. The theorem then follows for large enough $\codebookBlocklength$ by choosing $\finalconstTwo$ positive, but smaller than the minimum of these three positive terms, and $\finalconstTwo < \proofconstantOne$.
\end{proof}

\begin{proof}[Proof of Theorem~\ref{theorem:soft-covering-two-transmitters-second-order}]
We consider the typical sets $\typicalSetIndex{\typicalityParam_1}{\codebookBlocklength}{1}$ and $\typicalSetIndex{\typicalityParam_2}{\codebookBlocklength}{2}$, where for $\txIndex=1,2$, we choose $\typicalityParam_\txIndex >0$ to be
\begin{align}
\label{proof:soft-covering-two-transmitters-second-order-typicalityparam}
\typicalityParam_\txIndex
:=
\sqrt{\frac{\channelDispersion{\txIndex}}
           {\codebookBlocklength}
}
\normalcdfComplementInverse(\typicalityParam)
+
\secondOrderParamD
\frac{\log \codebookBlocklength}{\codebookBlocklength}.
\end{align}
The definitions~(\ref{def:soft-covering-atypical-term-one}), (\ref{def:soft-covering-atypical-term-two}) and (\ref{def:soft-covering-typical-term}) change accordingly.

In order to bound $\totvarAtypicalOne$, we use Theorem~\ref{theorem:berry-esseen} to obtain
\showto{arxiv}{\pagebreak}
\begin{align*}
&\phantom{{}={}}
\Probability_{\channelInOne^\codebookBlocklength, \channelInTwo^\codebookBlocklength, \channelOut^\codebookBlocklength}((\channelInOne^\codebookBlocklength, \channelInTwo^\codebookBlocklength, \channelOut^\codebookBlocklength) \notin \typicalSetIndex{\typicalityParam_1}{\codebookBlocklength}{1})
\\
&
=
\Probability_{\channelInOne^\codebookBlocklength, \channelInTwo^\codebookBlocklength, \channelOut^\codebookBlocklength}\left(
  \frac{1}{\codebookBlocklength}
  \sum\limits_{\blockIndex=1}^\codebookBlocklength
  \left(
    \informationDensityConditional{\channelInOne_\blockIndex}{\channelOut_\blockIndex}{\channelInTwo_\blockIndex}
    -
    \mutualInformationConditional{\channelInOne}{\channelOut}{\channelInTwo}
  \right)
  >
  \typicalityParam_1
\right)
\\
&\leq
\normalcdfComplement\left(
  \typicalityParam_1
  \sqrt{\frac{\codebookBlocklength}
             {\channelDispersion{1}}
  }
\right)
+
\frac{\channelThirdMoment{1}}
     {\channelDispersion{1}^{\frac{3}{2}} \sqrt{\codebookBlocklength}}
=
\secondOrderAtypicalProbability{1}.
\end{align*}
An application of Lemma~\ref{lemma:soft-covering-two-transmitters-atypical} with $\lemmaconst = 1/\sqrt{\codebookBlocklength}$ yields
\begin{multline}
\label{proof:soft-covering-two-transmitters-second-order-atypical-bound-1}
\Probability_{\codebookOne, \codebookTwo} \left(
  \totvarAtypicalOne
  >
  \secondOrderAtypicalProbability{1}\left(
    1+\frac{1}{\sqrt{\codebookBlocklength}}
  \right)
\right)
\\
\leq
\exp\left(
  -\frac{2\secondOrderAtypicalProbability{1}^2}
        {\codebookBlocklength}
  \exp(\codebookBlocklength \min(\codebookRateOne,\codebookRateTwo))
\right).
\end{multline}
Reasoning along similar lines shows
\begin{align*}
\Probability_{\channelInTwo^\codebookBlocklength, \channelOut^\codebookBlocklength}((\channelInTwo^\codebookBlocklength, \channelOut^\codebookBlocklength) \notin \typicalSetIndex{\typicalityParam_2}{\codebookBlocklength}{2})
\leq
\secondOrderAtypicalProbability{2}
\end{align*}
so that a renewed application of Lemma~\ref{lemma:soft-covering-two-transmitters-atypical} gives
\begin{multline}
\label{proof:soft-covering-two-transmitters-second-order-atypical-bound-2}
\Probability_{\codebookOne, \codebookTwo}\left(
  \totvarAtypicalTwo
  >
  \secondOrderAtypicalProbability{2}\left(
    1+\frac{1}{\sqrt{\codebookBlocklength}}
  \right)
\right)
\\
\leq
\exp\left(
  -\frac{2\secondOrderAtypicalProbability{2}^2}
        {\codebookBlocklength}
  \exp(\codebookBlocklength \min(\codebookRateOne,\codebookRateTwo))
\right).
\end{multline}
For the typical term, we use the definitions~(\ref{def:soft-covering-typical-term-one}) and~(\ref{def:soft-covering-good-codebooks}) with the typical set $\typicalSetIndex{\typicalityParam_1}{\codebookBlocklength}{1}$, and observe that for any fixed $\channelInTwoAlphElement^\codebookBlocklength$ and $\channelOutAlphElement^\codebookBlocklength$, we can apply Lemma~\ref{lemma:soft-covering-two-transmitters-typical} with $\generalrvOne=\channelInOne$, $\generalrvTwo=\channelInTwo$, $\generalrvThree=\channelOut$ and $\lemmaconst=1/\sqrt{\codebookBlocklength}$ to obtain
\showto{conference}{\pagebreak}
\begin{multline}
\label{proof:soft-covering-two-transmitters-second-order-typical-bound}
\Probability_{\codebookOne}\left(
  \totvarTypicalOne{\channelInTwoAlphElement^\codebookBlocklength}{\channelOutAlphElement^\codebookBlocklength}
  >
  1 + \frac{1}{\sqrt{\codebookBlocklength}})
\right) \\
\leq
\exp\left(
  -\frac{1}{3\codebookBlocklength} \exp(-\codebookBlocklength (\mutualInformationConditional{\channelInOne}{\channelOut}{\channelInTwo} + \typicalityParam_1 - \codebookRateOne))
\right).
\end{multline}

Now proceeding in a similar manner as in~(\ref{proof:soft-covering-two-transmitters-totalprob1}) to~(\ref{proof:soft-covering-two-transmitters-lemmaapplication2}) shows
\begin{multline*}
\Probability_{\codebookOne, \codebookTwo}\left(
  \totvarTypical{\channelOutAlphElement^\codebookBlocklength}
  >
  1 + \frac{3}{\sqrt{\codebookBlocklength}}
  ~|~
  \codebookOne \in \codebookSet_{\channelOutAlphElement^\codebookBlocklength}
\right)
\\
\leq
\exp\left(
  -\frac{1}{3\codebookBlocklength} \exp(-\codebookBlocklength (\mutualInformation{\channelInTwo}{\channelOut} + \typicalityParam_2 - \codebookRateTwo))
\right),
\end{multline*}
where there is no assumption on $\codebookBlocklength$ because $1/\sqrt{\codebookBlocklength} \leq 1$ for all $\codebookBlocklength \geq 1$.

The theorem then follows from (\ref{proof:soft-covering-two-transmitters-second-order-union-bound-start}) to (\ref{proof:soft-covering-two-transmitters-second-order-union-bound-end}) in the appendix,
\showto{arxiv}{
  \vfill\null
  \pagebreak
  \noindent
}
where~(\ref{proof:soft-covering-two-transmitters-second-order-union-bound-application}) results from~(\ref{proof:soft-covering-two-transmitters-typical-split}) and the union bound, (\ref{proof:soft-covering-two-transmitters-second-order-union-bound-substitutions}) follows by substituting (\ref{proof:soft-covering-two-transmitters-second-order-atypical-bound-1}), (\ref{proof:soft-covering-two-transmitters-second-order-atypical-bound-2}) and~(\ref{proof:soft-covering-two-transmitters-second-order-typical-bound}), and (\ref{proof:soft-covering-two-transmitters-second-order-union-bound-end}) follows by substituting (\ref{theorem:soft-covering-two-transmitters-second-order-rate-one}), (\ref{theorem:soft-covering-two-transmitters-second-order-rate-two}) and (\ref{proof:soft-covering-two-transmitters-second-order-typicalityparam}), as well as elementary operations.
\end{proof}

\bibliographystyle{plain}
\bibliography{references}

\clearpage
\onecolumn
\appendix

\begin{align}
\label{proof:soft-covering-two-transmitters-probability-bound-start}
&\phantom{{}={}}
\Probability_{\channelInOne^\codebookBlocklength, \channelInTwo^\codebookBlocklength, \channelOut^\codebookBlocklength}((\channelInOne^\codebookBlocklength, \channelInTwo^\codebookBlocklength, \channelOut^\codebookBlocklength) \notin \typicalSetIndex{\typicalityParam}{\codebookBlocklength}{1})
=
\Probability_{\channelInOne^\codebookBlocklength, \channelInTwo^\codebookBlocklength, \channelOut^\codebookBlocklength}\left(
  \vphantom{
    \frac{\channelpmf_{\channelOut^\codebookBlocklength | \channelInOne^\codebookBlocklength, \channelInTwo^\codebookBlocklength}(\channelOut^\codebookBlocklength | \channelInOne^\codebookBlocklength, \channelInTwo^\codebookBlocklength)}
        {\channelpmf_{\channelOut^\codebookBlocklength | \channelInTwo^\codebookBlocklength}(\channelOut^\codebookBlocklength | \channelInTwo^\codebookBlocklength)}
  }
  \right.
  \frac{\channelpmf_{\channelOut^\codebookBlocklength | \channelInOne^\codebookBlocklength, \channelInTwo^\codebookBlocklength}(\channelOut^\codebookBlocklength | \channelInOne^\codebookBlocklength, \channelInTwo^\codebookBlocklength)}
       {\channelpmf_{\channelOut^\codebookBlocklength | \channelInTwo^\codebookBlocklength}(\channelOut^\codebookBlocklength | \channelInTwo^\codebookBlocklength)}
  >
  \left.
  \vphantom{
    \frac{\channelpmf_{\channelOut^\codebookBlocklength | \channelInOne^\codebookBlocklength, \channelInTwo^\codebookBlocklength}(\channelOut^\codebookBlocklength | \channelInOne^\codebookBlocklength, \channelInTwo^\codebookBlocklength)}
        {\channelpmf_{\channelOut^\codebookBlocklength | \channelInTwo^\codebookBlocklength}(\channelOut^\codebookBlocklength | \channelInTwo^\codebookBlocklength)}
  }
  \exp(\codebookBlocklength(\mutualInformationConditional{\channelInOne}{\channelOut}{\channelInTwo}+\typicalityParam))
\right)
\\
&=
\Probability_{\channelInOne^\codebookBlocklength, \channelInTwo^\codebookBlocklength, \channelOut^\codebookBlocklength}
\left(
\vphantom{
  \frac{\channelpmf_{\channelOut^\codebookBlocklength | \channelInOne^\codebookBlocklength, \channelInTwo^\codebookBlocklength}(\channelOut^\codebookBlocklength | \channelInOne^\codebookBlocklength, \channelInTwo^\codebookBlocklength)}
          {\channelpmf_{\channelOut^\codebookBlocklength | \channelInTwo^\codebookBlocklength}(\channelOut^\codebookBlocklength | \channelInTwo^\codebookBlocklength)}
}
\right.
  \left(
    \frac{\channelpmf_{\channelOut^\codebookBlocklength | \channelInOne^\codebookBlocklength, \channelInTwo^\codebookBlocklength}(\channelOut^\codebookBlocklength | \channelInOne^\codebookBlocklength, \channelInTwo^\codebookBlocklength)}
        {\channelpmf_{\channelOut^\codebookBlocklength | \channelInTwo^\codebookBlocklength}(\channelOut^\codebookBlocklength | \channelInTwo^\codebookBlocklength)}
  \right)^{\renyiParam-1}
  >
  \exp(\codebookBlocklength(\renyiParam-1)(\mutualInformationConditional{\channelInOne}{\channelOut}{\channelInTwo}+\typicalityParam))
\left.
\vphantom{
  \frac{\channelpmf_{\channelOut^\codebookBlocklength | \channelInOne^\codebookBlocklength, \channelInTwo^\codebookBlocklength}(\channelOut^\codebookBlocklength | \channelInOne^\codebookBlocklength, \channelInTwo^\codebookBlocklength)}
       {\channelpmf_{\channelOut^\codebookBlocklength | \channelInTwo^\codebookBlocklength}(\channelOut^\codebookBlocklength | \channelInTwo^\codebookBlocklength)}
}
\right)
\\
&\leq
\Expectation_{\channelInOne^\codebookBlocklength, \channelInTwo^\codebookBlocklength, \channelOut^\codebookBlocklength}\left(
  \left(
    \frac{\channelpmf_{\channelOut^\codebookBlocklength | \channelInOne^\codebookBlocklength, \channelInTwo^\codebookBlocklength}(\channelOut^\codebookBlocklength | \channelInOne^\codebookBlocklength, \channelInTwo^\codebookBlocklength)}
        {\channelpmf_{\channelOut^\codebookBlocklength | \channelInTwo^\codebookBlocklength}(\channelOut^\codebookBlocklength | \channelInTwo^\codebookBlocklength)}
  \right)^{\renyiParam-1}
\right)
\cdot \exp(-\codebookBlocklength(\renyiParam-1)(\mutualInformationConditional{\channelInOne}{\channelOut}{\channelInTwo}+\typicalityParam))
\\
&=
\exp(
  \codebookBlocklength(\renyiParam-1)
  \cdot (
    \renyidiv{\renyiParam}{\Probability_{\channelInOne, \channelInTwo, \channelOut}}{\Probability_{\channelInOne | \channelInTwo}\Probability_{\channelOut | \channelInTwo}\Probability_\channelInTwo}
    -
    \mutualInformationConditional{\channelInOne}{\channelOut}{\channelInTwo}-\typicalityParam 
  )
)
\leq
\exp(-\codebookBlocklength\proofconstantOne)
\label{proof:soft-covering-two-transmitters-probability-bound-end}
\end{align}

\hrule

\begin{align}
\label{proof:soft-covering-two-transmitters-totalprob1}
&
\begin{aligned}
  \tilde{\Probability}
  =
  \sum\limits_{\hat{\codebookTwo}}
    \Probability_{\codebookTwo}(\codebookTwo = \hat{\codebookTwo})
    \Probability_{\codebookOne, \codebookTwo}\left(
    \vphantom{\sum\limits_{\codewordIndex_1=1}^{\exp(\codebookBlocklength\codebookRateOne)}}
    \right.
      &\sum\limits_{\codewordIndex_2=1}^{\exp(\codebookBlocklength\codebookRateTwo)}
        \exp(-\codebookBlocklength\codebookRateTwo)
        \frac{\channelpmf_{\channelOut^\codebookBlocklength | \channelInTwo^\codebookBlocklength}(\channelOutAlphElement^\codebookBlocklength | \codebookTwoWord{\codewordIndex_2})}
            {\channelpmf_{\channelOut^\codebookBlocklength}(\channelOutAlphElement^\codebookBlocklength)}
        \indicator{(\codebookTwoWord{\codewordIndex_2}, \channelOutAlphElement^\codebookBlocklength) \in \typicalSetIndex{\typicalityParam}{\codebookBlocklength}{2}}
        \\
        &\left. \vphantom{\sum\limits_{\codewordIndex_1=1}^{\exp(\codebookBlocklength\codebookRateOne)}}
        \totvarTypicalOne{\codebookTwoWord{\codewordIndex_2}}{\channelOutAlphElement^\codebookBlocklength}
        >
        1 + 3\exp(-\codebookBlocklength\proofconstantOne)
        ~|~
        \codebookOne \in \codebookSet_{\channelOutAlphElement^\codebookBlocklength}, \codebookTwo = \hat{\codebookTwo}
  \right)
\end{aligned}
\\
&
\begin{aligned}
  \phantom{\tilde{\Probability}}
  \leq
  \label{proof:soft-covering-two-transmitters-lemmaapplication1}
  \sum\limits_{\hat{\codebookTwo}}
    \Probability_{\codebookTwo}(\codebookTwo = \hat{\codebookTwo})
    \Probability_{\codebookOne, \codebookTwo}\left(
      \sum\limits_{\codewordIndex_2=1}^{\exp(\codebookBlocklength\codebookRateTwo)}
        \exp(-\codebookBlocklength\codebookRateTwo)
        \frac{\channelpmf_{\channelOut^\codebookBlocklength | \channelInTwo^\codebookBlocklength}(\channelOutAlphElement^\codebookBlocklength | \codebookTwoWord{\codewordIndex_2})}
            {\channelpmf_{\channelOut^\codebookBlocklength}(\channelOutAlphElement^\codebookBlocklength)}
        \indicator{(\codebookTwoWord{\codewordIndex_2}, \channelOutAlphElement^\codebookBlocklength) \in \typicalSetIndex{\typicalityParam}{\codebookBlocklength}{2}}
        \right. \\
      \left. \vphantom{\sum\limits_{\codewordIndex_1=1}^{\exp(\codebookBlocklength\codebookRateOne)}}
      >
      \frac{1 + 3\exp(-\codebookBlocklength\proofconstantOne)}
          {1 + \exp(-\codebookBlocklength\proofconstantOne)}
      ~|~
      \codebookOne \in \codebookSet_{\channelOutAlphElement^\codebookBlocklength}, \codebookTwo = \hat{\codebookTwo}
    \right)
\end{aligned}
\\
&\phantom{\tilde{\Probability}}\leq
\label{proof:soft-covering-two-transmitters-totalprob2-and-bound}
\Probability_{\codebookTwo}\left(
  \sum\limits_{\codewordIndex_2=1}^{\exp(\codebookBlocklength\codebookRateTwo)}
    \exp(-\codebookBlocklength\codebookRateTwo)
    \frac{\channelpmf_{\channelOut^\codebookBlocklength | \channelInTwo^\codebookBlocklength}(\channelOutAlphElement^\codebookBlocklength | \codebookTwoWord{\codewordIndex_2})}
          {\channelpmf_{\channelOut^\codebookBlocklength}(\channelOutAlphElement^\codebookBlocklength)}
    \indicator{(\codebookTwoWord{\codewordIndex_2}, \channelOutAlphElement^\codebookBlocklength) \in \typicalSetIndex{\typicalityParam}{\codebookBlocklength}{2}}
  >
  1 + \exp(-\codebookBlocklength\proofconstantOne)
\right)
\\
&\phantom{\tilde{\Probability}}\leq
\label{proof:soft-covering-two-transmitters-lemmaapplication2}
\exp\left(
  -\frac{1}{3} \exp(-\codebookBlocklength (\mutualInformation{\channelInTwo}{\channelOut} + \typicalityParam + 2\proofconstantOne - \codebookRateTwo))
\right)
\end{align}

\hrule

\begin{align}
\label{proof:soft-covering-two-transmitters-union-bound-start}
&\phantom{=}
\Probability_{\codebookOne, \codebookTwo} \left(
  \totalvariation{ \codebookpmf_{\channelOut^\codebookBlocklength | \codebookOne, \codebookTwo} - \channelpmf_{\channelOut^\codebookBlocklength}}
  >
  7\exp(-\codebookBlocklength\proofconstantOne)
\right)
\\
\label{proof:soft-covering-two-transmitters-union-bound-application}
&
\begin{aligned}
  \leq
  &\Probability_{\codebookOne, \codebookTwo}\left(
    \totvarAtypicalOne
    >
    2\exp(-\codebookBlocklength\proofconstantOne)
  \right)
  +
  \Probability_{\codebookOne, \codebookTwo}\left(
    \totvarAtypicalTwo
    >
    2\exp(-\codebookBlocklength\proofconstantOne)
  \right)
  \\
  &+
  \sum\limits_{\channelOutAlphElement^\codebookBlocklength \in \channelOutAlph^\codebookBlocklength} \left(
    \Probability_{\codebookOne}\left(
      \codebookOne \notin \codebookSet_{\channelOutAlphElement^\codebookBlocklength}
    \right)
    +
    \Probability_{\codebookOne, \codebookTwo}\left(
      \totvarTypical{\channelOutAlphElement^\codebookBlocklength}
      >
      1 + 3\exp(-\codebookBlocklength\proofconstantOne)
      ~|~
      \codebookOne \in \codebookSet_{\channelOutAlphElement^\codebookBlocklength}
    \right)
  \right)
\end{aligned}
\\
\label{proof:soft-covering-two-transmitters-union-bound-substitutions}
&
\begin{aligned}
  \leq
  &2\exp(-2\exp(\codebookBlocklength(\min(\codebookRateOne,\codebookRateTwo)-2\proofconstantOne)))
  +
  \cardinality{\channelOutAlph}^\codebookBlocklength
  \cardinality{\channelInTwoAlph}^\codebookBlocklength
  \exp\left(
    -\frac{1}{3} \exp(-\codebookBlocklength (\mutualInformationConditional{\channelInOne}{\channelOut}{\channelInTwo} + \typicalityParam + 2\proofconstantOne - \codebookRateOne))
  \right) \\
  &+
  \cardinality{\channelOutAlph}^\codebookBlocklength
  \exp\left(
    -\frac{1}{3} \exp(-\codebookBlocklength (\mutualInformation{\channelInTwo}{\channelOut} + \typicalityParam + 2\proofconstantOne - \codebookRateTwo))
  \right)
\end{aligned}
\end{align}

\hrule

\begin{align}
\label{proof:soft-covering-two-transmitters-second-order-union-bound-start}
&\phantom{=}
\Probability_{\codebookOne, \codebookTwo} \left( \totalvariation{ \codebookpmf_{\channelOut^\codebookBlocklength | \codebookOne, \codebookTwo} - \channelpmf_{\channelOut^\codebookBlocklength}}
>
(\secondOrderAtypicalProbability{2} + \secondOrderAtypicalProbability{1})
\left(1+\frac{1}{\sqrt{\codebookBlocklength}}\right)
+
\frac{3}{\sqrt{\codebookBlocklength}}
\right) \\
\label{proof:soft-covering-two-transmitters-second-order-union-bound-application}
&
\begin{aligned}
  \leq
  &\Probability_{\codebookOne, \codebookTwo}\left(
    \totvarAtypicalOne > \secondOrderAtypicalProbability{1}\left(1+\frac{1}{\sqrt{\codebookBlocklength}}\right)
  \right)
  +
  \Probability_{\codebookOne, \codebookTwo}\left(
    \totvarAtypicalTwo > \secondOrderAtypicalProbability{2}\left(1+\frac{1}{\sqrt{\codebookBlocklength}}\right)
  \right)
  \\
  &+
  \sum\limits_{\channelOutAlphElement^\codebookBlocklength \in \channelOutAlph^\codebookBlocklength} \left(
    \Probability_{\codebookOne}\left(
      \codebookOne \notin \codebookSet_{\channelOutAlphElement^\codebookBlocklength}
    \right)
    +
    \Probability_{\codebookOne, \codebookTwo}\left(
      \totvarTypicalOne{\codebookTwoWord{\codewordIndex_2}}{\channelOutAlphElement^\codebookBlocklength}
      >
      1 + \frac{3}{\sqrt{\codebookBlocklength}}
      ~|~
      \codebookOne \in \codebookSet_{\channelOutAlphElement^\codebookBlocklength}
    \right)
  \right)
\end{aligned}
\\
\label{proof:soft-covering-two-transmitters-second-order-union-bound-substitutions}
&
\begin{aligned}
  \leq
  &\exp\left(
    -\frac{2\secondOrderAtypicalProbability{2}^2}
          {\codebookBlocklength}
    \exp(\codebookBlocklength \min(\codebookRateOne,\codebookRateTwo))
  \right)
  +
  \exp\left(
    -\frac{2\secondOrderAtypicalProbability{2}^2}
          {\codebookBlocklength}
    \exp(\codebookBlocklength \min(\codebookRateOne,\codebookRateTwo))
  \right)
  \\
  &+
  \cardinality{\channelInTwoAlph}^\codebookBlocklength \cardinality{\channelOutAlph}^\codebookBlocklength
  \exp\left(
    -\frac{1}{3\codebookBlocklength} \exp(-\codebookBlocklength (\mutualInformationConditional{\channelInOne}{\channelOut}{\channelInTwo} + \typicalityParam_1 - \codebookRateOne))
  \right)
  +
  \cardinality{\channelOutAlph}^\codebookBlocklength
  \exp\left(
    -\frac{1}{3\codebookBlocklength} \exp(-\codebookBlocklength (\mutualInformation{\channelInTwo}{\channelOut} + \typicalityParam_2 - \codebookRateTwo))
  \right)
\end{aligned}  
\\
\label{proof:soft-covering-two-transmitters-second-order-union-bound-end}
&\leq
2\exp\left(
  -\frac{2\min(\secondOrderAtypicalProbability{1}^2,\secondOrderAtypicalProbability{2}^2)}
        {\codebookBlocklength}
  \exp(\codebookBlocklength \min(\codebookRateOne,\codebookRateTwo))
\right)
+
2\exp\left(
  \codebookBlocklength(\log \cardinality{\channelOutAlph} + \log \cardinality{\channelInTwoAlph})
  -\frac{1}{3}
  \codebookBlocklength^{\secondOrderParamC - \secondOrderParamD - 1}
\right)
\end{align}
\end{document}